\documentclass[sigconf]{acmart}
\usepackage{fancyhdr}
 \usepackage{amsmath}
\usepackage{cases}
\usepackage{stfloats}
\usepackage{enumerate}
\usepackage{url}
\usepackage{graphicx}
\usepackage{balance}
\usepackage{algorithm}
\usepackage{algpseudocode}

 \usepackage{etoolbox}

\newtoggle{ACM}
\toggletrue{ACM}

\newtoggle{IEEEcls}
\toggletrue{IEEEcls}
\togglefalse{IEEEcls}

\floatname{algorithm}{Algorithm}

\iftoggle{IEEEcls}{
\usepackage{amsthm}
\theoremstyle{plain}
\newtheorem{theorem}{Theorem}

\theoremstyle{definition}
\newtheorem{definition}{Definition}
 
\newtheorem{lemma}{Lemma}

}

\newif\ifNotUse  
 \NotUsetrue

 \iftoggle{ACM}{
\setcopyright{none}  
\acmDOI{}
\acmPrice{}
\acmISBN{}
\acmConference[ ]{ }{ }{ }
\acmYear{2018}

}

\begin{document}
\title{Economical and efficient network super points detection based on GPU}
\author{Jie Xu}
\affiliation{%
  \institution{School of computer science and engineer, Southeast university}
  \city{Nanjing}
  \state{China}
 }
\email{xujieip@163.com}

\iftoggle{IEEEcls}{
\author{\IEEEauthorblockN{Jie Xu                                                
                                  }
\IEEEauthorblockA{School of Computer Science and Engineering\\ South East University \\
Nanjing, China\\
Email: xujieip@163.com}

}

}

\begin{abstract}
Network super point is a kind of special host which plays an important role in network management and security. For a core network, detecting super points in real time is a burden task because it requires plenty computing resources to keep up with the high speed of packets. Previous works try to solve this problem by using expensive memory, such as static random access memory, and multi cores of CPU. But the number of cores in CPU is small and each core of CPU has a high price. In this work, we use a popular parallel computing platform, graphic processing unit GPU, to mining core network's super point. We propose a double direction hash functions group which can map hosts randomly and restore them from a dense structure. Because the high randomness and simple process of the double direction hash functions, our algorithm reduce the memory to smaller than one-fourth of other algorithms. Because the small memory requirement of our algorithm, a low cost GPU, only worth 200 dollars, is fast enough to deal with a high speed network such as 750 Gb/s. No other algorithm can cope with such a high bandwidth traffic as accuracy as our algorithm on such a cheap platform. Experiments on the traffic collecting from a core network demonstrate the advantage of our efficient algorithm.
\end{abstract} 
 
  \iftoggle{IEEEcls}{
}


\maketitle
\section{Introduction}

Host's cardinality, the number of other hosts contacting with it during a time window, is a significant attribute in network security and management\cite{Instr:AnIntrusionDetectionPreventionSystemCloudComputing}\cite{Instr:AsurveyIntrusionDetectionTechniquesInCloud:ChiragModi}. Super point is a host whose cardinality is bigger than a predefined threshold. It plays important roles in many network events, for example scanners\cite{Instruction:AcomprehensiveReview}, DDos attackers or victims, servers, instruction detection\cite{Instr:AnEfficientIntrusionDetectionSystemBasedSupportVectorMachines:LiYinHui} and so on.

Network scanning is one of the most common forms of network intrusion and often a prelude to a more damaging attack \cite{Scan:NetworkScanDetectionWithLQS}. Attackers use network scan to map the topology of a target network and identify active hosts running interesting network services. When scanning started, thousands of packets sending to different hosts would transmit from attacker's host, which let the attacker's host become a super host \cite{Scan:EvasionResistantNetworkScanDetection}. By super hosts detection and monitor, network scan can be detected effectively and we can prevent a future attacker.

Distributed Denial of Service (DDos) attacks, which caused tremendous economic losses every year, is a big threat for network security \cite{DDos:AttackingDDoSAtSource}. Attackers, often hiding in a Botnet, flood huge faked packets with different source IP addresses to a server which would cause the normal users failed to visit the server or the server crashed down immediately. Under DDos attacks, a server's cardinality would be much bigger than normal \cite{DDos:ImplementingPushbackRouter} \cite{DDos:ATaxonomyOfDDoSAttack}. By monitoring the cardinality fluctuation of a super point, we can detect a DDos attack immediately. 

Finding super hosts on core network is a challenge because there are massive IP packets passing in the traffic every second. If we had infinite memory and fast enough accessing speed, we could store each IP address in a hash table and worked out the cardinality accurately by storing them in memory one by one. In fact, this method would be too expensive and slow to apply to core network whose speed is more than 40 Gb/s. Estimating method, compromise between memory consumption and detection accuracy, is widely used to solve this problem\cite{HSD:streamingAlgorithmFastDetectionSuperspreaders}\cite{HSD:ADataStreamingMethodMonitorHostConnectionDegreeHighSpeed}. 

Traffic speed could be measured from packets speed perspective(million packets per second Mpps) or size perspective(gigabit per second Gb/s). Algorithms will detect super point according to IP addresses extracting from IP packets. So the packets speed has more influence. 

According to report on Caida \cite{expdata:Caida}, IP packet's average size is 800 Bytes. So a 40Gb/s network will transfer 6 million packets every second. In another word, its packet speed is 6 Mpps. But the fastest super point algorithm \cite{HSD:LineSpeedAccurateSuperspreaderIdentificationDynamicErrorCompensation} running on CPU can only scanning 2 million packets in theory base on fast and very expensive memory, static random memory SRAM. So it's no possible to detect super points in real time for a network whose bandwidth is higher than 13.3 Gb/s. The bottleneck of the speed is memory latency and packets scanning speed. But the frequency of processor can not increase easily any more and the processing ability of single core is limited.

In order to deal with high speed network in real time, parallel computing devices are essential. Although nowadays CPU contains many cores than before, from 2 cores to 22 cores or even more, a multi-cores CPU is very expensive and the advanced mainboard , which cost more than 2000 dollars, can only install 4 CPU at most. This limits the the parallel computing ability of CPU platform.

Graphic processing unit GPU is a famous parallel computing platform widely used in many areas. Unlike CPU, a common GPU contains hundreds or thousands of physical processing cores and has lower price per core than that of CPU. For example, a 200 dollars GPU, Nvidia GTX 950, contains 640 cores. While a 400 dollars CPU, Intel i7-7700K, only contains 4 cores. What's more, a mainboard can install several GPUs and GPUs can communicate with each other directly too. GPU has the potential ability to detection super points economically and efficiently with a suitable algorithm. In this work, we devise a novel algorithm which can be deployed on a cheap GPU to deal with a core network traffic. Our algorithm is memory efficient and fast enough to process 120 million packets per second because a novel double direction hash functions group are designed to reconstruct super points from a dense memory structure.

In the next section we will introduce existed super points detection methods and other related works. In section 3, our novel super points detection algorithm will be discussed in detail. Section 4 describes how to deploy our algorithm on GPU. Experiments on real-world core network traffic are shown in section 5. We make a conclusion in section 6.
\section{Related work}
Many scholars have proposed several super points detection algorithms. Shobha et al.\cite{HSD:streamingAlgorithmFastDetectionSuperspreaders} proposed a scaleable algorithm which does not need to keep every host's state. Cao et al.\cite{HSD:identifyHighCardinalityHosts} used a pair-based sampling method to eliminate the majority of low opposite number hosts and reserved more resource to estimate the opposite number of the resting hosts. Estan et al.\cite{HSD:bitmapCountingActiveFlowsHighSpeedLinks} proposed two bits map algorithms based on sampling flows. But all of these methods were based on sampling flows which limited their accuracy. 

Wang et al.\cite{HSD:ADataStreamingMethodMonitorHostConnectionDegreeHighSpeed} devised a novel structure, called double connection degree sketch (DCDS), to store and estimate different hosts cardinalities. They updated DCDS by setting several bits to one simply. In order to restore super points at the end of a time period, which bits to be updated were determined by Chinese Remainder Theory(CRT) when parsing a packet. By using CRT, every bit of DCDS could be shared by different hosts. But the computing process of CRT was very complex which limited the speed of this algorithm.

Liu et al.\cite{HSD:DetectionSuperpointsVectorBloomFilter} proposed a simple method to restore super hosts basing on bloom filter. They called this algorithm as Vector Bloom Filter Algorithm(VBFA). VBFA used the bits extracted from IP address to decide which bits to be updated when scanning a packet. Compared with CRT, bit extraction only needed a small operation. But the bits extraction from an IP address were not random enough to spread hosts uniformly in VBF and the memory usage of this algorithm was very low.

Most of the previous works only focused on accelerating speed by adapting fast memory but neglected the calculation ability of processors. Seon-Ho et al.\cite{HSD:GPU:2014:AGrandSpreadEstimatorUsingGPU} first used GPU to estimate hosts opposite numbers. They devised a Collision-tolerant hash table to filter flows from origin traffic and used a bitmap data structure to record and estimate hosts' opposite numbers. But this method needed to store IP address of every flow while scanning traffic because they could not restore super points from the bitmap directly. Additional candidate IP address storing space increased the memory requirement of this algorithm.

In this paper, we design a novel algorithm which can restore super points efficiently without keeping a IP addresses list while scanning packets. 
\section{Super point detection algorithm}
In this section we will introduce our super points detection and estimation algorithm based on novel double direction hash functions group. Host's cardinality calculation is the foundation of estimating algorithm. So we first introduce the cardinality estimation algorithm used in our algorithm.

\subsection{Cardinality Estimation}
Suppose there are two networks: $NT_a$ and $NT_b$. Let $A$ and $B$ represent the hosts set of these two networks. $NT_a$ and $NT_b$ communicate with each other by a group of edge routers as shown in figure \ref{TwoNetworkModel}. These routers compose the edge of $NT_a$ and $NT_b$, written as $E(a,b)$. Through $E(a,b)$ we can observe packets stream in two directions: from $NT_a$ to $NT_b$ and from $NT_b$ to $NT_a$. For a host $a_0$ in $NT_a$, the task of cardinality estimation is to get the number of hosts in $NT_b$ that communicate with $a_0$ in a certain time period. "communicate" means send packets to or receive packets from $a_0$. Let $CO(a_0)$ represent the set of these hosts. Calculate $a_0$'s cardinality is to find the number of elements in $CO(a_0)$, written as $|CO(a_0)|$. Because a host in $CO(a_0)$ may have several packets communicating with $a_0$, we should find out the number of distinct hosts from packets stream. Super point is a host whose cardinality is more than a certain threshold $\theta$. After getting hosts' cardinalities, we can filter super points easily.

\begin{figure}[!ht]
\centering
\includegraphics[width=0.47\textwidth]{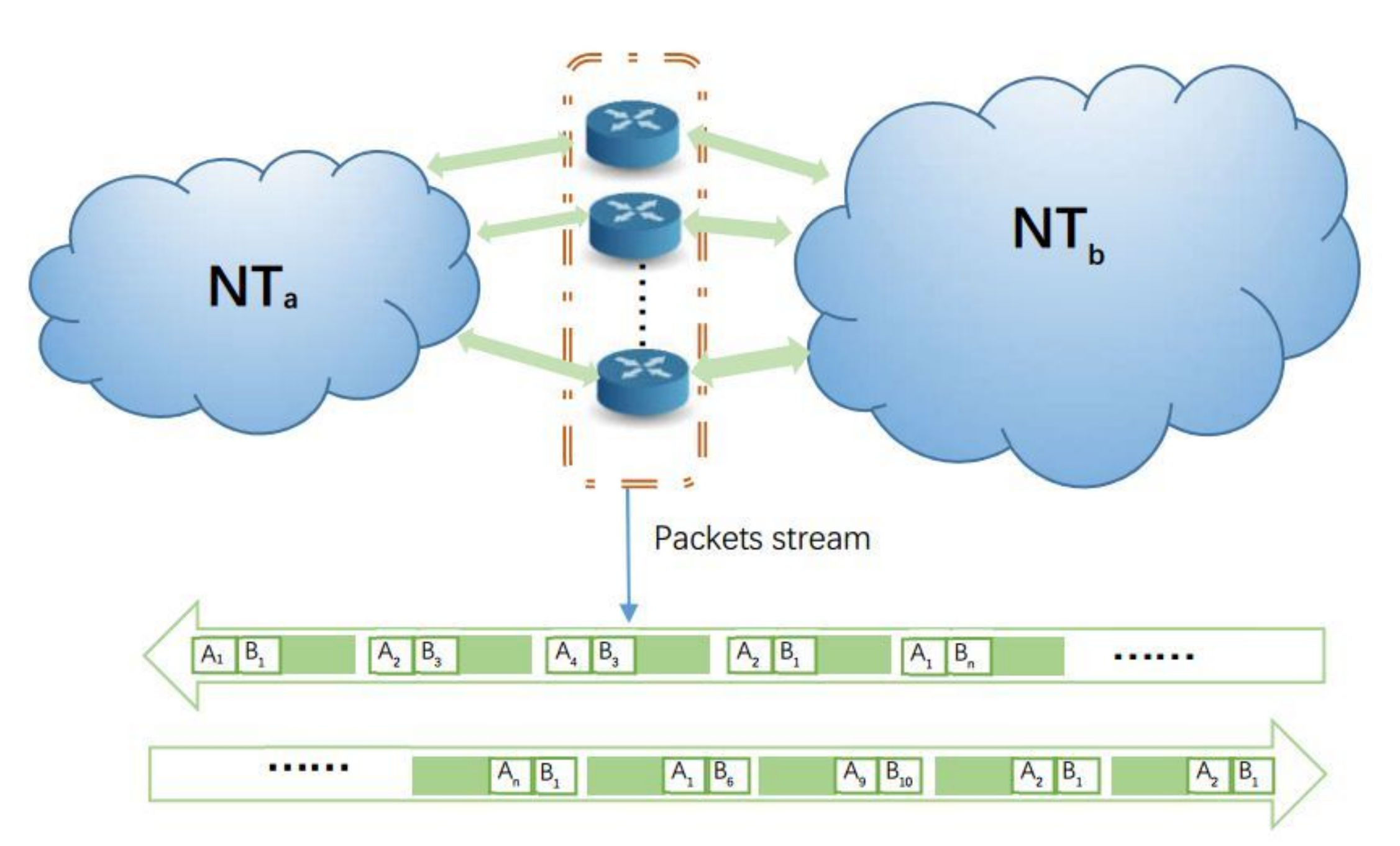}
\caption{Edge of two networks}
\label{TwoNetworkModel}
\end{figure}

WHANG et al \cite{DC:aLinearTimeProbabilisticCountingDatabaseApp} gives a simple way  to estimate the distinct number from data stream, called linear estimator. We use linear estimator to calculate host's cardinality. At the beginning, a bit vector with $g$ zero bits is allocated. When scanning a packet with a host $b_0$ in $NT_b$ and $a_0$ in $NT_a$ at $E(a,b)$, we set a random bit to 1. This bit is select by $b_0$'s hashed value with a random hash function $H_1$\cite{hash_UniversalOne-wayHashFunctionsAndTheirCryptographicApplications}. After scanning all packets in a time period, we can report the estimating cardinality $Est'$ by the follow equation:
\begin{equation}\label{eqt_LDC_estValueFromZeroBitsN}
 {Est}'=-g*ln(\frac{z_0}{g})
\end{equation}
$z_0$ is the number of zero bits remaining at the end of the time period. We call this algorithm as linear estimator.

Linear estimator uses fixed number of bits and has a fast processing speed because for every packet it only needs to access memory once. But in a high speed network, it's too expensive to allocate a linear estimator for every host because of the following reasons.

\begin{enumerate}
\item Memory requirement. There are huge hosts in high speed network. Allocating a linear estimator for every host requires great memory. For example, when $g$ is set to 1024, and there are 8 million hosts in $NT_a$, 1 GB memory is desired. Allocating or transmitting such huge memory between different edge routers is a heavy burden. 
\item Linear estimator locating. These hosts in a network may have discrete addresses and we can't locate their linear estimators directly by their IP addresses. In order to find a host's linear estimator, we have to use hash table or balance tree to store these linear estimators. But hash table has collision problems\cite{hashTable:HTree:AnEfficientIndexStructureEventMatchingContentBasedPublish}  and balance tree requires additional memory accessing\cite{Cormen:2009:IAT:IntroductionAlgorithmsThirdEdition}.
\end{enumerate}

To solve this problem, we design a novel structure which can detect super points and estimate their cardinalities using fixed number of linear estimators. We will introduce our novel algorithm in detail in the following part.
\subsection{Scanning packets stream}
Because it is low efficient that a linear estimator is used by a host exclusive, we let a linear estimator can be shared by several hosts. Let $LA^{g,k}$ represent an array of $2^k$ estimators with $g$ bits in each one. Every host will select a linear estimator from $LA^{g,k}$ to estimator its cardinality. 

But this will cause a higher estimation value than the real cardinality. In order to reduce the effect of sharing linear estimator, we use $r$ $LA^{g,k}$s together and each host will select $r$ linear estimators from every $LA^{g,k}$ to record its cardinality at the same time as described in figure \ref{LinearEstimatorArrays_model}

\begin{figure}[!ht]
\centering
\includegraphics[width=0.47\textwidth]{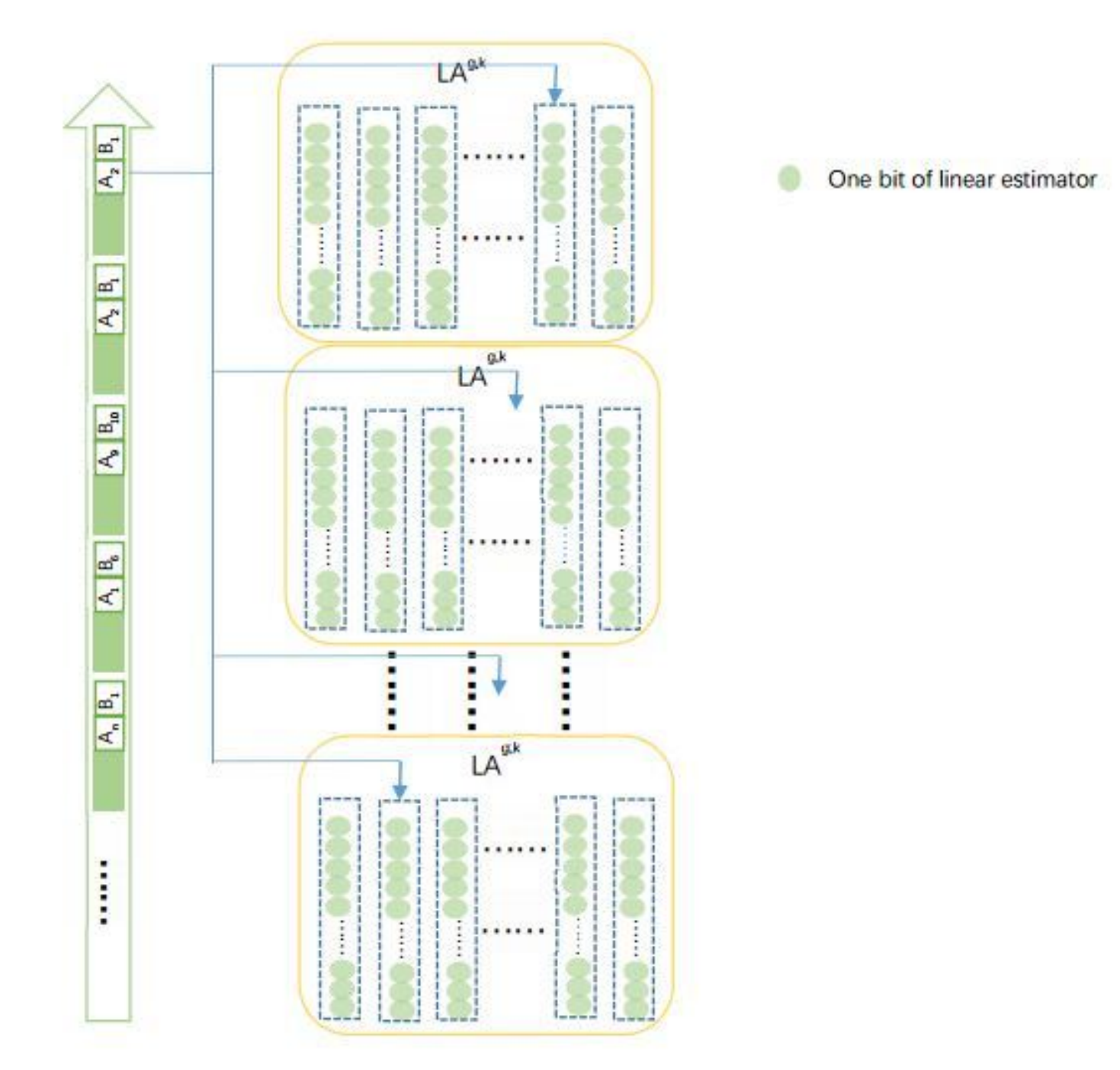}
\caption{Linear estimator arrays}
\label{LinearEstimatorArrays_model}
\end{figure}

For a host $a_0$ in $NT_a$, let $LE(a_0)$ represent the set linear estimators in these $r$ $LA^{g,k}$s that used to record $a_0$'s cardinality. At the end of a time period, we will calculate the union linear estimator $ULE(a_0)$ by applying bit-wise "AND" operation to all linear estimators in $LE(a_0)$ and get $a_0$'s estimating value from $ULE(a_0)$.

In order to get a high accuracy estimation value and detect super points from these $r$ $LA$, how to select linear estimators in different $LA^{g,k}$ should comply with the following two requirements:

\begin{enumerate}
\item Each estimator in $LA$ should have the same probability to be used by different hosts. This will make sure that every linear estimator is fully utilized and the algorithm can acquire the highest accuracy with the smallest memory.
\item For a super point $a_0$, we can reconstruct it from $LE(a_0)$. At the end of a time period, all what we have are $r$ $LA^{g,k}$s that record all hosts' cardinalities. Super points will be reconstructed from all of these $LA^{g,k}$s.
\end{enumerate}

The key to these requirements is how to map a host to $r$ linear estimators in $LA$( in the rest of this paper, we use $LA$ to represent $LA^{g,k}$ simply). We devise a novel double direction hash functions group, written as $DHG$, which contains $r$ random hash functions. $DHG$ can hash an integer $i$ to $r$ random values between 0 and $2^k-1$ where $i$ and $k$ are both positive integers. What's more, $i$ could be restored from its hashed $r$ values. We can use equation \ref{eqt_LDC_estValueFromZeroBitsN} to find out which linear estimator's estimating value is more than $\theta$ and we call this kind of estimator as hot estimator. When $a_0$ is a super point, $LE(a_0)$ will appear in hot estimators. By $DHG$, we can restore $a_0$ from $LE(a_0)$.
 
The first hash function $DH_0$ in $DHG$ is a random hash function\cite{hash_UniversalOne-wayHashFunctionsAndTheirCryptographicApplications} that maps a host $a_0$ to a value between 0 and $2^k-1$ uniformly. $DH_0$ has excellent randomness, but it requires great operations and is not reversible. For the sake of super point reconstruction and randomness, the rest $r-1$ hash functions are determined by the following equation:
\begin{equation}\label{eqt_DHG_hashFuntions}
DH_i(a_0)=((a_0>>(i-1)*\alpha))mod (2^k)) \bigoplus DH_0(a_0) , 1 \leq i \leq r-1
\end{equation}
$\alpha$ is a positive integer, "$>>$"is bitwise right shift operation and $\bigoplus$ is a bitwise "exclusive or" operation. $DH_0(a_0)$ is a random seed which makes sure all the $r$ hash functions has a high randomness. Because $k$ is a positive integer, $(a_0>>(i-1)*\alpha))mod (2^k)$ is successive $k$ bits of $a_0$ starting from $(i-1)*\alpha$. We call these bits as bit block, written as $BL(i)$. $BL(i)$ could be restored from $DH_i(a_0)$ by equation
\ref{eqt_DHG_ip_part_reverise}.
\begin{equation}\label{eqt_DHG_ip_part_reverise}
BL(i) = DH_i(a_0) \bigoplus DH_0(a_0) , 1 \leq i \leq r-1
\end{equation}

If $BL(i)$ is long enough, we can reconstruct $a_0$ by concatenating these $r-1$ $BL$s. We call these $r$ $LA$s hashing by $DHG$ as Double direction Hash Linear estimators Array  $DHLA$. Algorithm \ref{updateDHLA} shows how to scan packets and update $DHLA$.

\begin{algorithm}                       
\caption{Update DHLA}          
\label{updateDHLA}                            
\begin{algorithmic}                    
    \Require {\\ IP address pair $<a_0,b_0>$,\\
     DHLA}     
\State $BITidx \Leftarrow H_1(b_0)$
\For{$i \in [0,r-1]$}
\State $LEidx \Leftarrow DH_i(a_0)$
\State $LEp$ point to the $LEidx$th $LE$ in the $i$ $LA$
\If {the $BITidx$th bit of $LEp$ is 1}
\State Continue
\Else
\State Set the $BITidx$th bit of $LEp$ to 1
\EndIf
\EndFor 
\State Return
\end{algorithmic}
\end{algorithm}

Algorithm \ref{updateDHLA} has very little operations which let it have a fast processing speed for core network which forwards millions of packets every second. $DHLA$ maintains all hosts' cardinalities and when $\alpha$ and $r$ is reasonably selected, super points could be restored efficiently. In the next part, how to restore super points from $DHLA$ will be discussed in detail.

\subsection{Restoring super points}
$BL(i)$ is a fraction of $a_0$. If every bit of $a_0$ is contained in some $BL(i)$, $a_0$ could be reconstructed entirely. Figure \ref{BitsBlockExtraction} illustrates the relationship between $a_0$ and $BL(i)$.

\begin{figure}[!ht]
\centering
\includegraphics[width=0.47\textwidth]{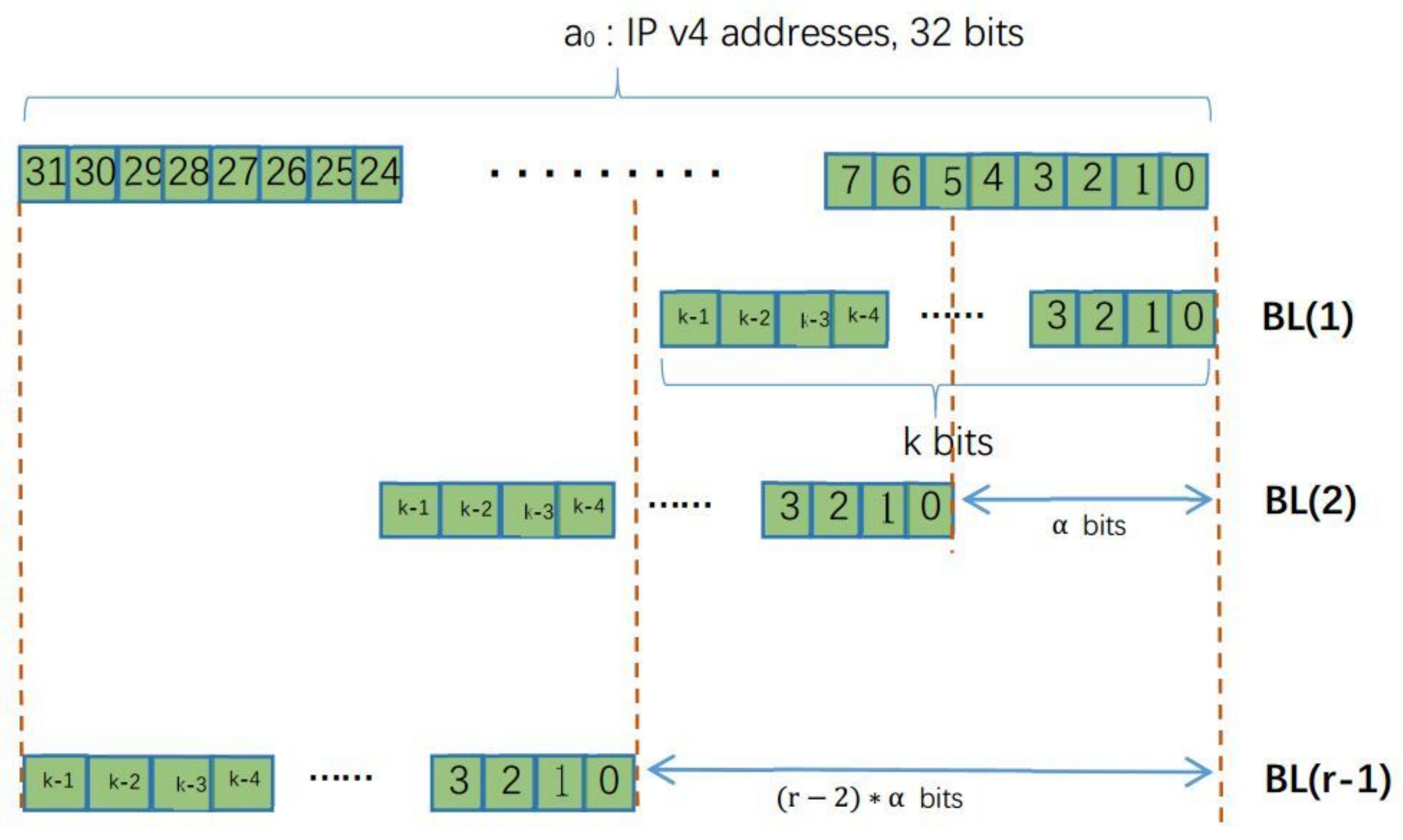}
\caption{Bit block extracted from a IP address}
\label{BitsBlockExtraction}
\end{figure}

In order to reconstruct $a_0$, $r$ $\alpha$ and $k$ should comply with the following two restrictions:
\begin{enumerate}
\item $\alpha \leq k$. This makes sure that there are no missing bits between two neighboring $BL$s. When $0 < \alpha < k$, there will be $k-\alpha$ bits be the same in two neighboring $BL$s. These duplicating bits will help us to filter candidate IP from $DHLA$.

\item $(r-2)*\alpha + k \geq 32$. This restriction means that the last $BL$ will contains the last bit of $a_0$.
\end{enumerate}
When these two conditions are met, $a_0$ could be restored from these $r-1$ $BL$s by comparing and concatenating operations and $BL(i)$ could be acquired by equation \ref{eqt_DHG_ip_part_reverise}. 

But do not store $LE(a_0)$ when scanning packets, we can't get $BL(i)$ directly. According to the definition of super points, when $a_0$ is a super point, each linear estimator in $LE(a_0)$ is a hot estimator. We can first find out all hot estimators in every $LA$ and then test them one by one to restore super points. After getting $a_0$, we can get $ULE(a_0)$ and estimate $a_0$'s cardinality according to the zero number in $ULE(a_0)$. But equation \ref{eqt_LDC_estValueFromZeroBitsN} could not be used here directly because some bits of $ULE(a_0)$ may be set by other hosts. In order to get a high accuracy result, we should estimate the number of "1" bits set by other hosts and remove them. The following lemma and theorem show how to do this. Flow in this paper means the set of packets with the same IP addresses.

\begin{lemma}
\label{la-1BitProbability}
In a certain time period, if there are $w$ flows in the total traffic, the probability that a bit in a $LA$ to be set to 1 is $\psi=1-e^{-\frac{w}{g*(2^k)}}$
\end{lemma} 
\begin{proof}
Because every host is mapped to a linear estimator in a $LA$ uniformly, a linear estimator will receive $\frac{w}{2^k}$ flows. By equation \ref{eqt_LDC_estValueFromZeroBitsN}, we can calculate the "0" bit number ${z}'$ of linear estimator which is ${z}'=g*e^{-\frac{w}{g*2^k}}$. Because ${z}'$ bits are uniform distributed in a linear estimator's $g$ bits, the probability of a bit being set to 1 is $\psi=\frac{g-{z}'}{g}=1-e^{-\frac{w}{g*2^k}}$.
\end{proof}

$w$ could be estimated by equation \ref{eqt_LDC_estValueFromZeroBitsN} from every $LA$ because $LA$ could be regarded as a bit set used to record the flow number. We can get the "0" bit number $ZR(i)$ of the $i$th $LA$, and estimate the flow number ${w(i)}'$ from it ${w(i)}'=-g*2^k*ln(\frac{ZR(i)}{g*2^k})$. We use the average value of all the $r$ estimating value of every $LA$ as the flow number ${w}'=\frac{\sum_{i=0}^{r-1}{w(i)}'}{r}$.

In $LA$, several hosts may be mapped to the same linear estimator. It causes that the "0" bit number of $ULE(a_0)$ is smaller than that when the linear estimators are used by $a_0$ exclusively.
\begin{definition}[Sharing zero number/Exclusive zero number]
\label{def-Sharing-ExclusiveZeroN}
For a host $a_0$, its sharing "0" bit number is the number of zero bit in $ULE(a_0)$, written as $SZ(a_0)$; its exclusive zero number is the number of zero bit of a linear estimator which only used by $a_0$, written as $EZ(a_0)$.
\end{definition}

\begin{theorem}
\label{Th-estimatingZN}
 For a host $a_0$, its $EZ$ can be acquired from $SZ(a_0)$ by equation $EZ(a_0)=\frac{SZ(a_0)}{1-{\psi}^{r}}$.
 \end{theorem}
\begin{proof}
Suppose a linear estimator only records $a_0$'s opposite hosts. At the end of a time window, the "1" bit number of this linear estimator is $g-EZ(a_0)$. For a bit in $ULE(a_0)$, if it is not set by $a_0$, it will be set to one with probability ${\psi}^{LR}$ because it must be set by some host in every linear estimator in $LE(a_0)$. The expecting "1" bit number of these $EZ(a_0)$ bits is $EZ(a_0)*{\psi}^{r}$. Then the "1" bits number is expected to be $g-EZ(a_0)+EZ(a_0)*{\psi}^{r}$. At the end of a time window, we get the number of "1" bit $g-SZ(a_0)$ by counting zero bits. Let the expecting value equal to the watched value $g-EZ(a_0)+EZ(a_0)*{\psi}^{LR}=k-SZ(a_0)$. Reform this equation, we have $EZ(a_0)=\frac{SZ(a_0)}{1-{\psi}^{r}}$.
\end{proof}

In theorem \ref{Th-estimatingZN}, $SZ(a_0)$ is the zero number in $ULE(a_0)$. $a_0$'s cardinality $C(a_0)$ could be estimated from $EZ(a_0)$ by the following equation:
\begin{equation}\label{eqt_cardinalityEstimating}
C(a_0)=-g*ln(\frac{SZ(a_0)}{g-g*{\psi}^{r}})
\end{equation}

 Algorithm \ref{restoreIP} shows how to detect super points from $DHLA$ and estimate their cardinalities.

\begin{algorithm}                       
\caption{restore super points}          
\label{restoreIP}                            
\begin{algorithmic}                    
\Require {\\ DHLA,\\$\theta$}  
\Ensure {Super points list $SPL$}
\State $Zmin\Leftarrow g*e{-\frac{\theta}{g}}$ 
\For { $i \in [0,r-1]$}
\For {$j \in [0,2^k-1]$}
  \State $z_0 \Leftarrow g-|LA(i,j)|$
   \If {$z_0<Zmin$}
   \State add $j$ into $HE(i)$
   \EndIf
\EndFor
\EndFor  
\For {each $<CL_0,CL_1,CL_2,\cdots,CL_{r-1}> \in <HE(0),HE(1),HE(2),\cdots,HE(r-1)>$}
\State ContianIP $\Leftarrow true$
\For {$i \in [1 , r-2]$}
\State $BL(i)\Leftarrow (CL_0 \bigoplus CL_i)$
\State $BL(i+1)\Leftarrow (CL_0 \bigoplus CL_{i+1})$
\If { the left $k - \alpha$ bits of $BL(i)$ not equal to the right $k-\alpha$ bits oof $BL(i+1)$}
\State ContianIP $\Leftarrow true$
\State Break
\EndIf
\EndFor
\If { ContianIP == false}
\State Continue
\EndIf
\State restore $a_0$ from $<CL_0,CL_1,CL_2,\cdots,CL_{r-1}>$
\State get $ULE(a_0)$
\State $C(a_0)\Leftarrow -g*ln(\frac{g-|ULE(a_0)|}{g-g*{\psi}^{r}})$
\If {$C(a_0) < \theta$}
\State Continue
\EndIf
\State insert $a_0$ and $C(a_0)$ to $SPL$
\EndFor 
\State Return $SPL$
\end{algorithmic}
\end{algorithm}

In algorithm \ref{restoreIP}, $LA(i,j)$ points to the $j$th linear estimator in the $i$ $LA$. The weight of $LA(i,j)$, written as $|LA(i,j)|$ is the number of "1"bits in it. 

Both algorithm \ref{updateDHLA} and \ref{restoreIP} have no data conflict which make sure that we can launch them on several threads concurrence to get a high processing speed. In the next section, we will introduce how to deploy our super points detection algorithm on GPU.
\section{Detect super point on GPU}
GPU is designed for graphic processing originally. It has plenty physic cores to deal with different pixels parallel. Because its power of computation, GPU is widely used in parallel computing as a popular platform. GPU contains more cores than CPU and has a higher memory accessing speed. When scanning different data with the same instructions (single instruction multi data, SIMD), GPU has much better performance than CPU does. When deploying on GPU, our algorithm can scan packets and restore super points parallel. At the begin of a time period, we allocate $DHLA$ on GPU global memory with every bit set to 0. Then GPU threads can update and read this $DHLA$ directly.

\subsection{Scanning packets on GPU}
From each edge router, we can monitor a packet stream and extract IP address pairs from this stream. Super points will be detected from IP address pairs. Packets could be mirrored to a monitor sever alongside a router. If we want to use GPU to scan IP address pairs, we have to copy them from monitor server's memory to GPU's global memory. An IP address pair is a vector of two IP addresses extracting from a packet. It is low efficient to copy every IP address pair one by one because the coping procedure requires additional starting and ending operations. We allocate a IP address buffer on both monitor server and GPU. This IP pair buffer could store $\mu$ IP pairs and its size is $8*\mu$ bytes. Figure \ref{GPU_scanningIPpairs} illustrates how to scan IP pairs on monitor server. When monitor server receives packets from router, it will extract IP pairs and store them to IP pair buffer on its memory. When the IP buffer has stored $\mu$ IP pairs, server monitor will copy it to GPU buffer. GPU card is locating on monitor server's board connecting it with PCIe3.0 bus. After copying to GPU, monitor server's buffer will be cleared for storing another $\mu$ IP pairs. When receiving $\mu$ IP pairs from monitor server, GPU will launch $\mu$ threads running algorithm \ref{updateDHLA} to process this $\mu$ IP pairs parallel. In algorithm \ref{updateDHLA}, a bit could be set by several threads at the same time without mistakes. So these $\mu$ threads could update $DHLA$ parallel. 

\begin{figure}[!ht]
\centering
\includegraphics[width=0.47\textwidth]{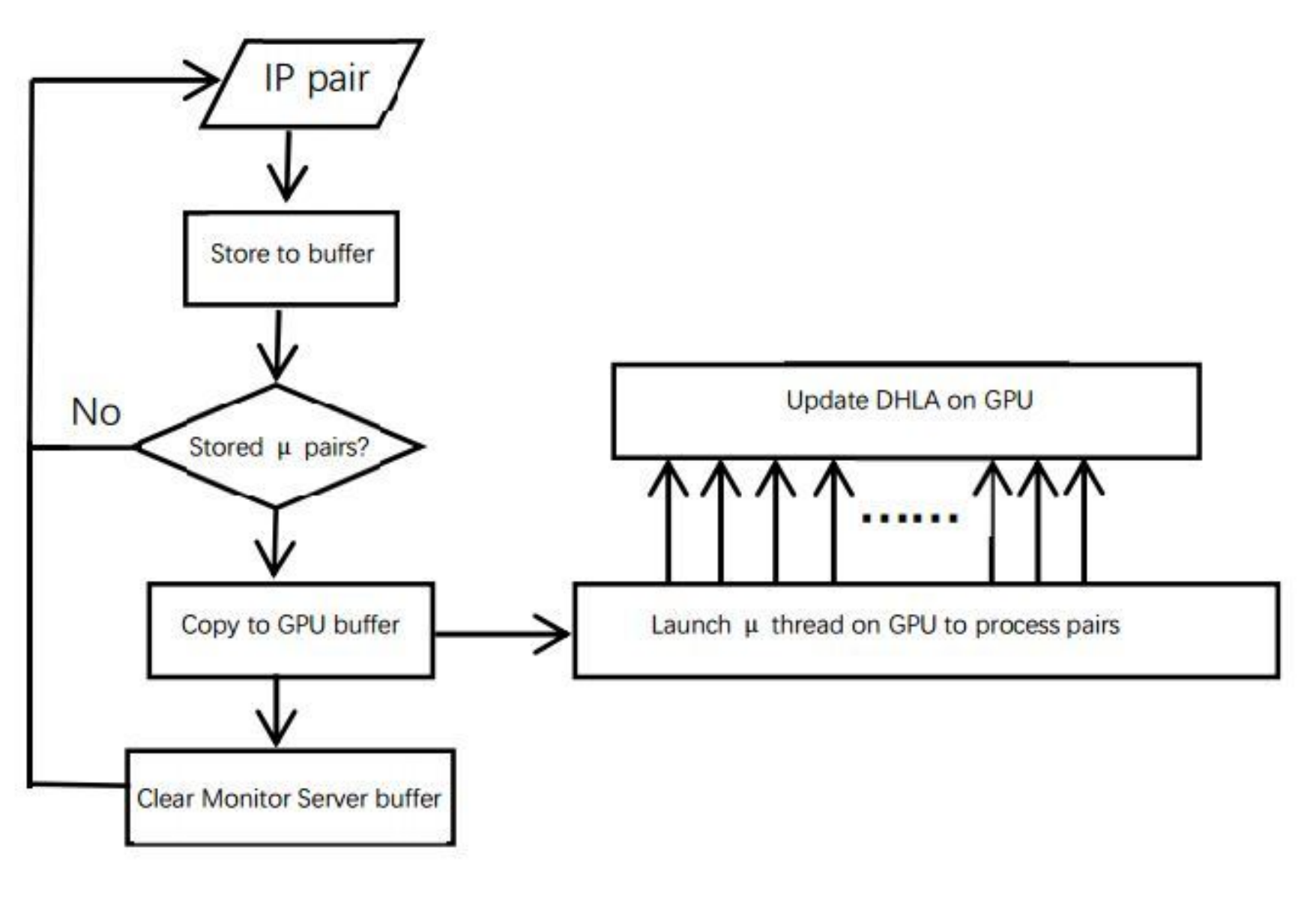}
\caption{IP pairs scanning and copying}
\label{GPU_scanningIPpairs}
\end{figure}
\subsection{super point cardinality estimation on GPU}
At the end of a time period, we can estimate super point cardinality from $DHLA$. If there are more than one edge routers, we can merge their local $DHLA$s to a global one by bit wise "AND" operation and use the global $DHLA$ to estimate super point cardinality. In the rest of this section, $DHLA$ means the global $DHLA$ merging from all local monitor servers.

Super point cardinality estimation could be divided into three parts: hot linear estimators identifying, super points detection, cardinality estimation. All of these parts could be implement on GPU parallel. 

To get $HE(i)$, we have to count the zero bit number of every linear estimator. There are total $r*2^k$ linear estimators in $DHLA$. The number of linear estimators may be more than the maximum of available threads on GPU. We use a fixed number of threads to test all of these linear estimators and every thread will test more than one linear estimators. The hot linear estimator checking algorithm is very simple, counting the number of zero bits and comparing the number with a positive integer. This will not consume much time in GPU.

After getting $HE(i)$ of every $LA$, we will try to restore super points from them. In algorithm \ref{restoreIP}, select $r$ linear estimators from every $HE(i)$ to generate a hot estimator tuple and test every of this tuple one by one. This method is simple but very low efficient. Because there are $\prod_{i=0}^{r}{|HE(i)|}$ tuples, checking all of these tuples one by one will consume much time, especially when there are many super points. 

Observing that if $CL_0 \in HE_0$, $CL_1 \in HE_1$, $CL_2 \in HE_2$ and the left $k-\alpha$ bits of $CL_0 \bigoplus CL_1$ not equal to the right $k-\alpha$ bits of $CL_0 \bigoplus CL_1$, there is no need to check any tuples that contain $CL_0$ $CL_1$ and $CL_2$. So we will restore super points incrementally from $HE(1)$ to $HE(r-1)$ and store the intermediate result. But if we want to run parallel in GPU, we have to allocate two buffers for intermediate result: one for reading and the other for storing. Because $BL(i)$ depends on $DH_0$, in the intermediate result we will store which linear estimator is used to reconstruct $BL$ and the restored partial IP address. The first three $HE$ have different process with the others, so we will discuss them separately.

Firstly, we allocate two buffers in GPU global memory, $SubIP_1$ and $SubIP_2$, to store partial IP addresses and $DH_0$. The size of these buffers, $|SubIP_1|$ or $|SubIP_2|$, is the number of partial IP addresses it stores. For the first three $HE$s, there are $|HE_0|*|HE_1|*|HE_2|$ tuples like $<CL_0,CL_1,CL_2>$ where $CL_i \in HE_i$. Denote the set of these tuples as $TP_{<0,1,2>}$. Suppose we launch $TN$ threads on GPU to deal with $TP_{<0,1,2>}$. Let $a$, $b$ be two non-negative integers that $a*TN+b=|HE_0|*|HE_1|*|HE_2|$. Each of the first $b$ thread will test $a+1$ tuples and each of the rest $TN-b$ thread will check $a$ tuples. Let $TP_{<0,1,2>}(i)$ represent the set of tuples in $TP_{<0,1,2>}$ to be scanned by the $i$th thread. Algorithm \ref{alg_GPU_BL123} shows how every threads launched in GPU to deal with $TP_{<0,1,2>}$.

\begin{algorithm}                       
\caption{scanning $HE_0$, $HE_1$, $HE_2$ on GPU}          
\label{alg_GPU_BL123}                            
\begin{algorithmic}                    
\Require {\\ $TP_{<0,1,2>}$ }  
\State $tid \Leftarrow$ thread ID 
\State get $TP_{<0,1,2>}(tid)$ from $TP_{<0,1,2>}$
\State $SSubIP \Leftarrow SubIP_1$
\For { $<CL_0,CL_1,CL_2> \in TP_{<0,1,2>}(tid) $}
\State $blk_1 \Leftarrow CL_0 \bigoplus CL_1$
\State $blk_2 \Leftarrow CL_0 \bigoplus CL_2$
\If {the left $k-\alpha$ bits of $blk_1$ not equal to the\\
   right $k-\alpha$ bits of $blk_2$}
\State Continue
\EndIf
\State $subIP\Leftarrow$ concatenate the left $\alpha$ bits of $blk_2$ to $blk_1$ 
\State Insert $<subIP, CL_0>$ into $SSubIP$
\EndFor  
\end{algorithmic}
\end{algorithm}

The roles of these two buffers may change in different stages. $SSubIP$ points to the buffer used for storing sub IP addresses in a stage and $RSubIP$ points to the buffer used for reading sub addresses. When scanning $TP_{<0,1,2>}$, $SubIP_1$ is used for storing sub IP addresses. After scanning $TP_{<0,1,2>}$, $SubIP_1$ will contain $|SubIP_1|$ candidate sub IP addresses. Together with $HE_3$, we can generate another tuples $TP_{<0,1,2,3>}$ like $<CL_0, subip, CL_3>$ where $CL_0 \in SubIP_1$, $subip \in SubIP_1$, $CL_3 \in HE_3$. When scanning $TP_{<0,1,2,3>}$, $SubIP_2$ will be used for storing sub IP addresses, and scanning candidate partial IP addresses will be gotten from $SubIP_1$. Like $TP_{<0,1,2>}$, each thread will scan a sub set of $TP_{<0,1,2,3>}$. The rest $HE$ will be processed like this. $SubIP_1$ and $SubIP_2$ will play as storing buffer alternatively. Algorithm \ref{alg_GPU_BL_after3} shows how to scan the rest $HE$

\begin{algorithm}                       
\caption{scanning HE(i) on GPU}          
\label{alg_GPU_BL_after3}                            
\begin{algorithmic}                    
\Require {\\ $TP_{<0,1,2,…,i>}$ 
\\Super point list $SPL$ \\ }  
\State $tid \Leftarrow$ thread ID 
\State get $TP_{<0,1,2,…,i>}(tid)$ from $TP_{<0,1,2,…,i>}$
\If{ $i$ is an even number}
\State $SSubIP \Leftarrow SubIP_1$
\Else
\State $SSubIP \Leftarrow SubIP_2$
\EndIf
\For { $<CL_0,subip,CL_i> \in TP_{<0,1,2,…,i>}(tid) $}
\State $blk_2 \Leftarrow CL_0 \bigoplus CL_2$
\If {the left $k-\alpha$ bits of $subip$ not equal to the\\
   right $k-\alpha$ bits of $blk_2$}
\State Continue
\EndIf
\State $sub_2\Leftarrow$ concatenate the left $\alpha$ bits of $blk_2$ to $subip$ 
\If{$i$ equal to $r-1$}
\State Insert $sub_2$ into $SPL$
\Else
\State Insert $<sub_2, CL_0>$ into $SSubIP$
\EndIf
\EndFor  
\end{algorithmic}
\end{algorithm}

By algorithm \ref{alg_GPU_BL_after3}, we will restore candidate super points when scanning the last $HE$. Candidate super points are stored in a list for further checking and cardinality estimation. When dealing with $SPL$, we still start fix number of threads and each thread scans a sub set of candidate super points in $SPL$. For a candidate super point $a_0$, a thread will first get $ULE(a_0)$ from $DHLA$ and count the zero bit number $z$ in it. Then $a_0$'s cardinality could be calculated by equation \ref{eqt_cardinalityEstimating}. If the estimating value is more than $\theta$, we will report it as a super point together with its estimating cardinality.

Because the high randomness of $DHG$, $DHLA$ requires much smaller memory than other algorithms do. So our algorithm can run on a cheap GPU to deal with high speed networks. In the next section we will give the real world traffic experiments of our algorithm comparing with several state-of-the-art algorithms.
\section{Experiments}
We use a real world traffic to evaluate the performance of our super points detection algorithm: Double direction Hash Super points detection Algorithm (DHSA). The traffic is OC192 downloading from Caida\cite{expdata:Caida}. This traffic contains one hour packets last from 13:00 to 14:00 on February 19, 2015. In our experiment, the time period is set to 5 minutes and the threshold for super point is 1024. Under this time period, the one-hour traffic is divided into 12 sub traffics and we will detect super points from them. Table \ref{tbl-trafficInf} shows the detail information of every sub traffic. 
\begin{table*}
\centering
\caption{Traffic information}
\label{tbl-trafficInf}
\begin{tabular}{c}                                                                                                                                                                                                                           
\centering
\includegraphics[width=0.7\textwidth]{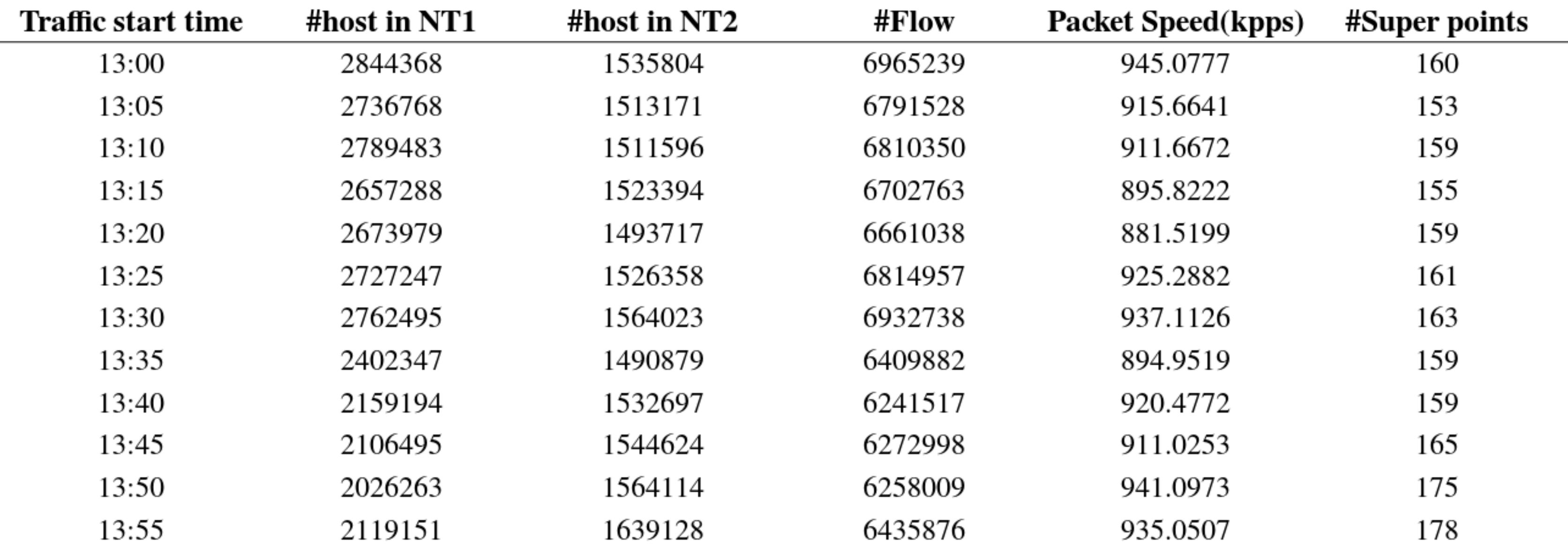}
\end{tabular}
\end{table*}

Accuracy, time consumption and memory requirement are three criteria to evaluate super point detection algorithm. False positive ratio FPR and false negative ratio FNR are two classic rates for accuracy comparing. They are given in definition \ref{def-fpr_fnr}.
\begin{definition}[FPR/FNR]
\label{def-fpr_fnr}
For a traffic with $N$ super points, an algorithm detects $N'$ super points. In the $N'$ detected super points, there are $N^+$ hosts which are not super points. And there are $N^-$ super points which are not detected by the algorithm. FPR means the ratio of $N^+$ to $N$ and FNR means the ratio of $N^-$ to $N$.
\end{definition}

FPR may decrease with the increase of FNR. If an algorithm reports more hosts as super point, its FNR will decrease but FPR will increase. So we use the sum of FPR and FNR, total false rate TFR, to evaluate the accuracy of an algorithm. 

The parameters of our algorithm DHSA, such as $r$ $g$ $k$, will affect its accuracy. Especially the parameter $g$, which also determines the accuracy of cardinality estimation, has great influence to DHSA. Figure \ref{DHSA_TFP_differentParameters} illustrates DHSA's TFR under different parameters when $g$ varies from 256 to 4096 on the first sub trafffic.

\begin{figure*}
\centering
\includegraphics[width=0.9\textwidth]{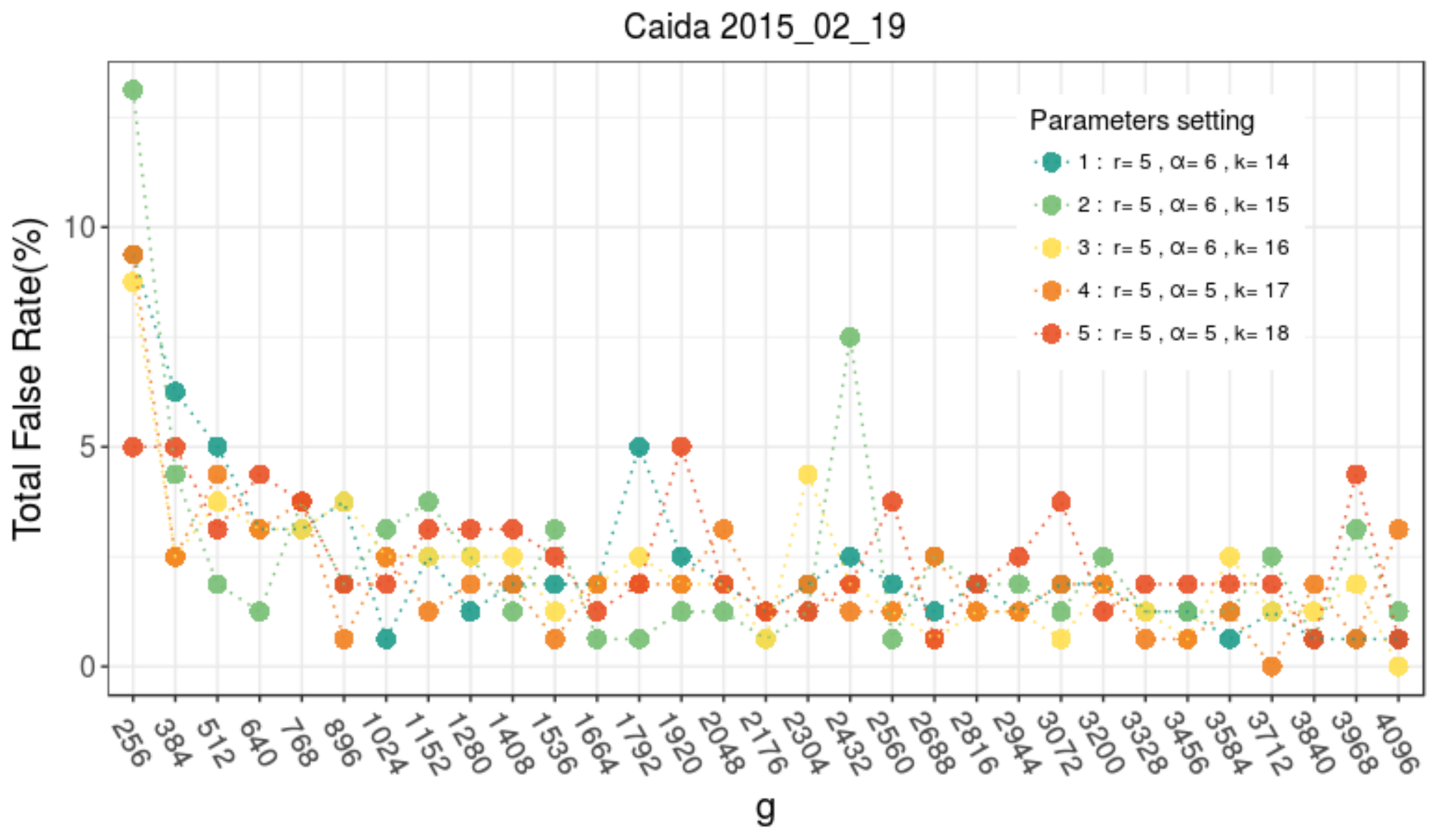}
\caption{TFR of DHSA under different parameters}
\label{DHSA_TFP_differentParameters}
\end{figure*}

TFR of DHSA is very high when $g$ is littler than 512. But it decreases slowly when $g$ increases after $1024$. When $g$ is set to 1024, DHSA can get an excellent result. In the rest of our experiments, we set $g=1024$, $r=5$, $\alpha=6$ and $k=14$.

To compare the performance of DHSA with other algorithms, we use DCDS\cite{HSD:ADataStreamingMethodMonitorHostConnectionDegreeHighSpeed}, VBFA\cite{HSD:DetectionSuperpointsVectorBloomFilter}, GSE \cite{HSD:GPU:2014:AGrandSpreadEstimatorUsingGPU} to compare with it. All of these algorithms are running on a common GPU card: GTX950 with 640 CUDA cores and 4 GB memory.

We compare the FPR, FNR and consuming time of these algorithms as shown in figure \ref{FPR_hsd_rlt_t12}, \ref{FNR_hsd_rlt_t12}, \ref{totalTime_hsd_rlt_t12}.

\begin{figure}[!ht]
\centering
\includegraphics[width=0.47\textwidth]{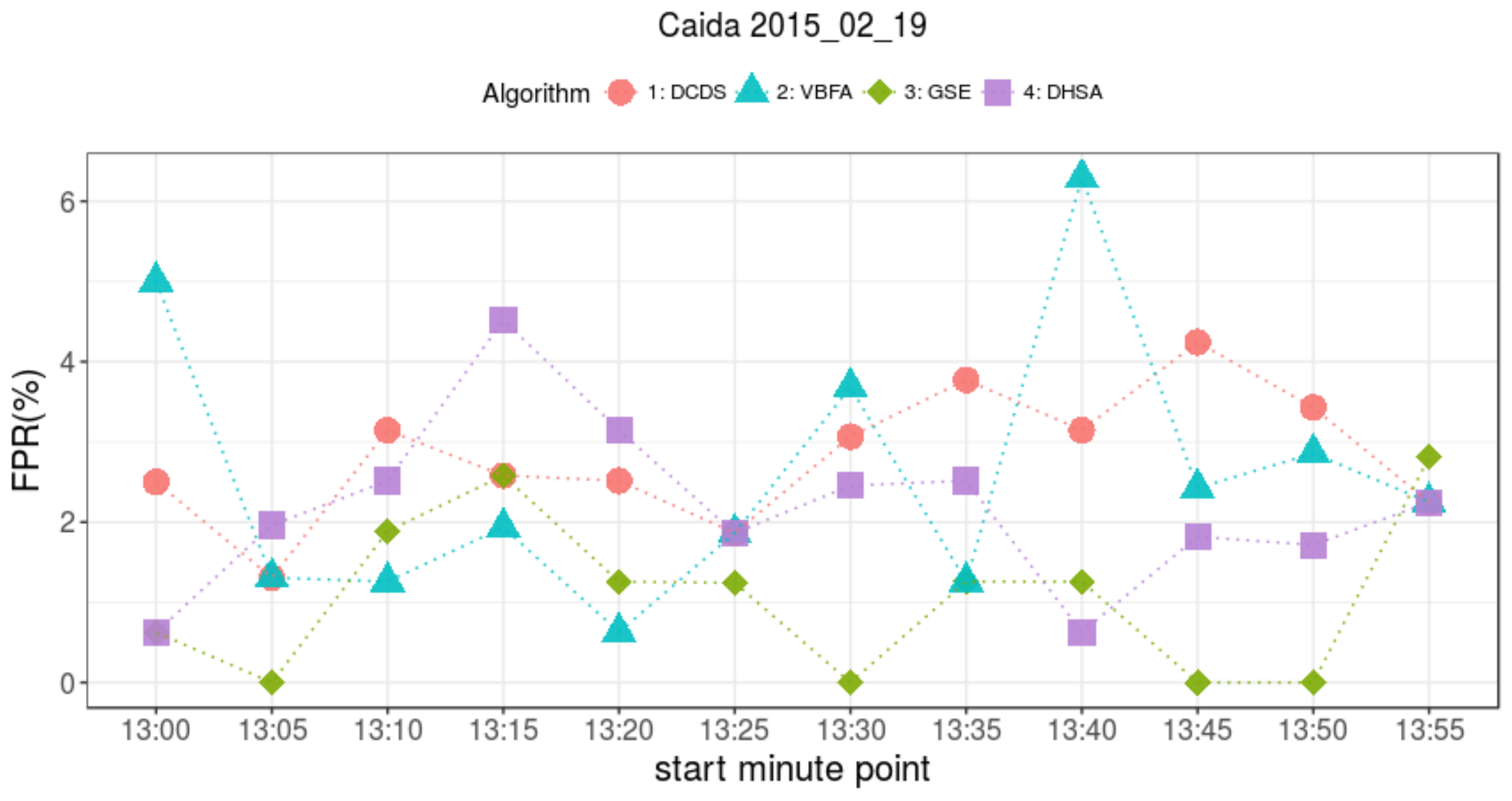}
\caption{False positive rate of different algorithms}
\label{FPR_hsd_rlt_t12}
\end{figure}

\begin{figure}[!ht]
\centering
\includegraphics[width=0.47\textwidth]{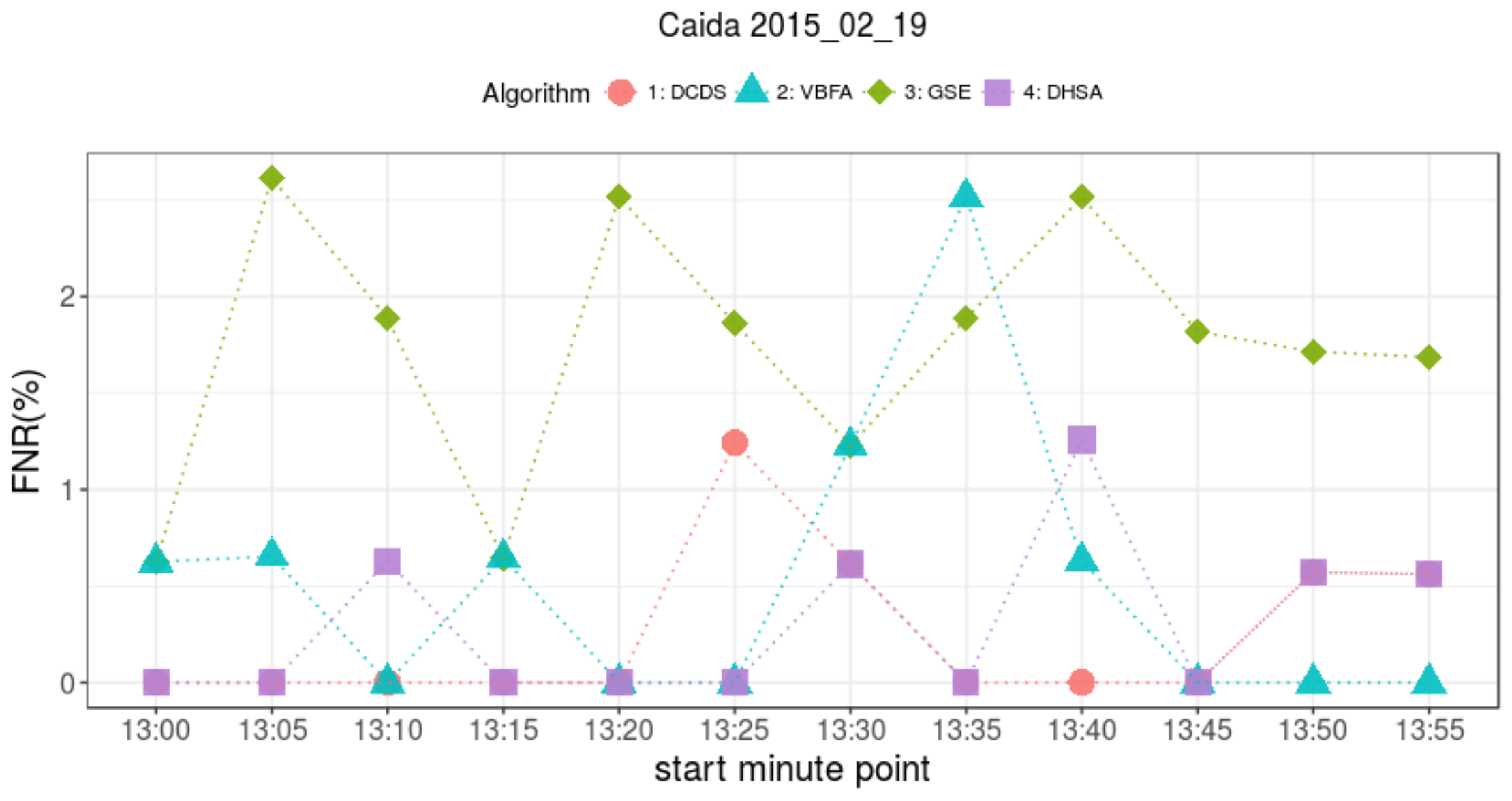}
\caption{False negative rate of different algorithms}
\label{FNR_hsd_rlt_t12}
\end{figure}

\begin{figure}[!ht]
\centering
\includegraphics[width=0.47\textwidth]{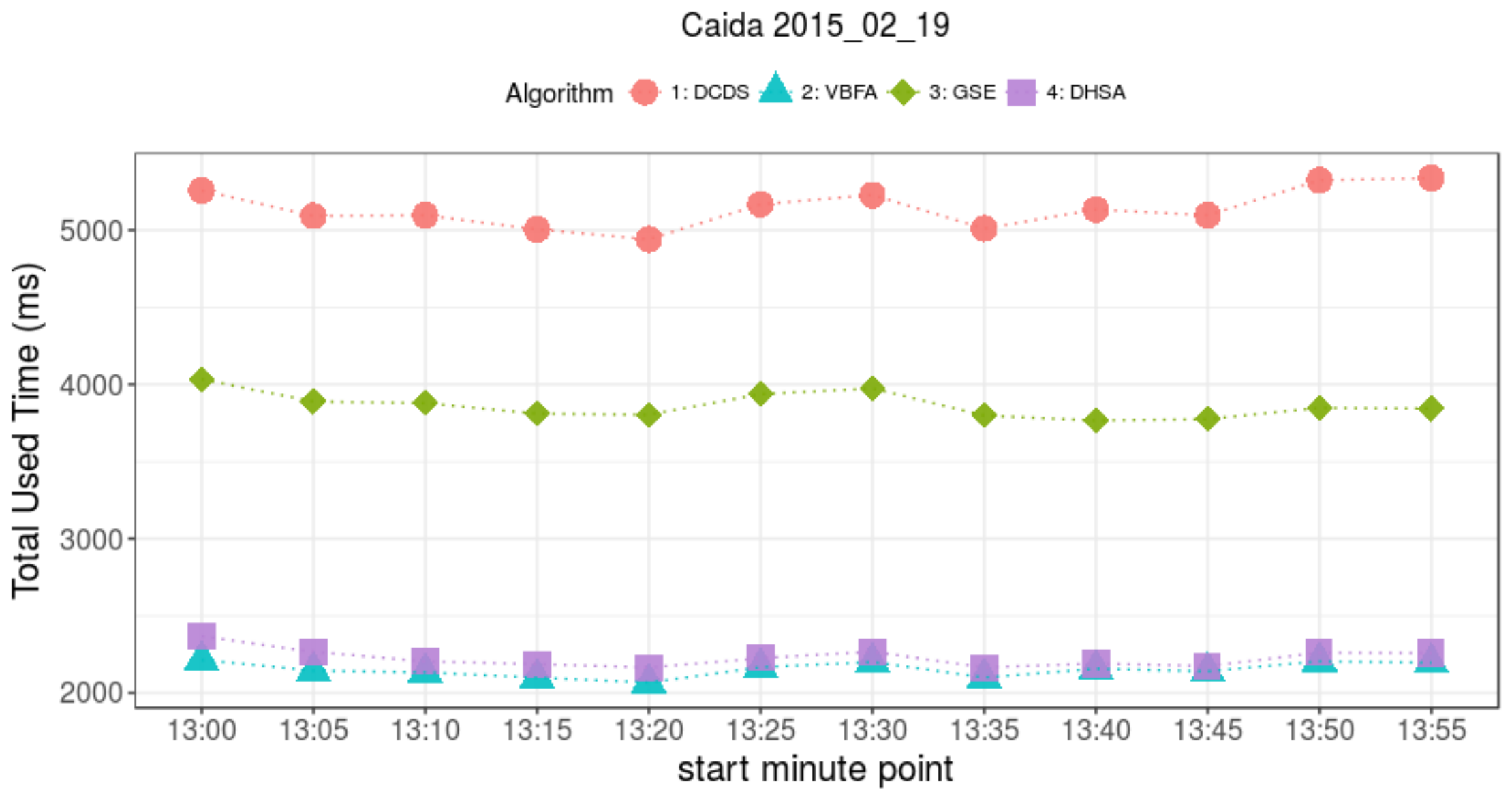}
\caption{Time consumption of different algorithms}
\label{totalTime_hsd_rlt_t12}
\end{figure}

GSE has a lower FPR than other algorithms. It can remove fake super points according the estimating flow number. But GSE may remove some super points too, which causes it has a higher FNR. Because it uses discrete bits to record host's cardinality, collecting all of these bits together when estimate super points cardinality will use lots of time. DCDS uses CRT when storing host's cardinality. CRT has a better randomness which makes DCDS has a lower FNR. But CRT is very complex containing many operations. So DCDS's speed is the lowest among all of these algorithms. DHSA and VBFA have lower FPR than DCDS and lower FNR than GSE. Time consumed by DHSA is only a little more than that of VBFA. But DHSA's FPR and FNR is lower than VBFA's because DHG makes full use of every linear estimator. What's more, DHSA consumes only one-fourth memory that used by VBFA. Table \ref{tbl_avg_hsd_rlt} lists the average result of all the 12 sub traffics.

\begin{table*}
\centering
\caption{Average detection result}
\label{tbl_avg_hsd_rlt}
\begin{tabular}{c}                                                                                                                                                                                                                           
\centering
\includegraphics[width=0.7\textwidth]{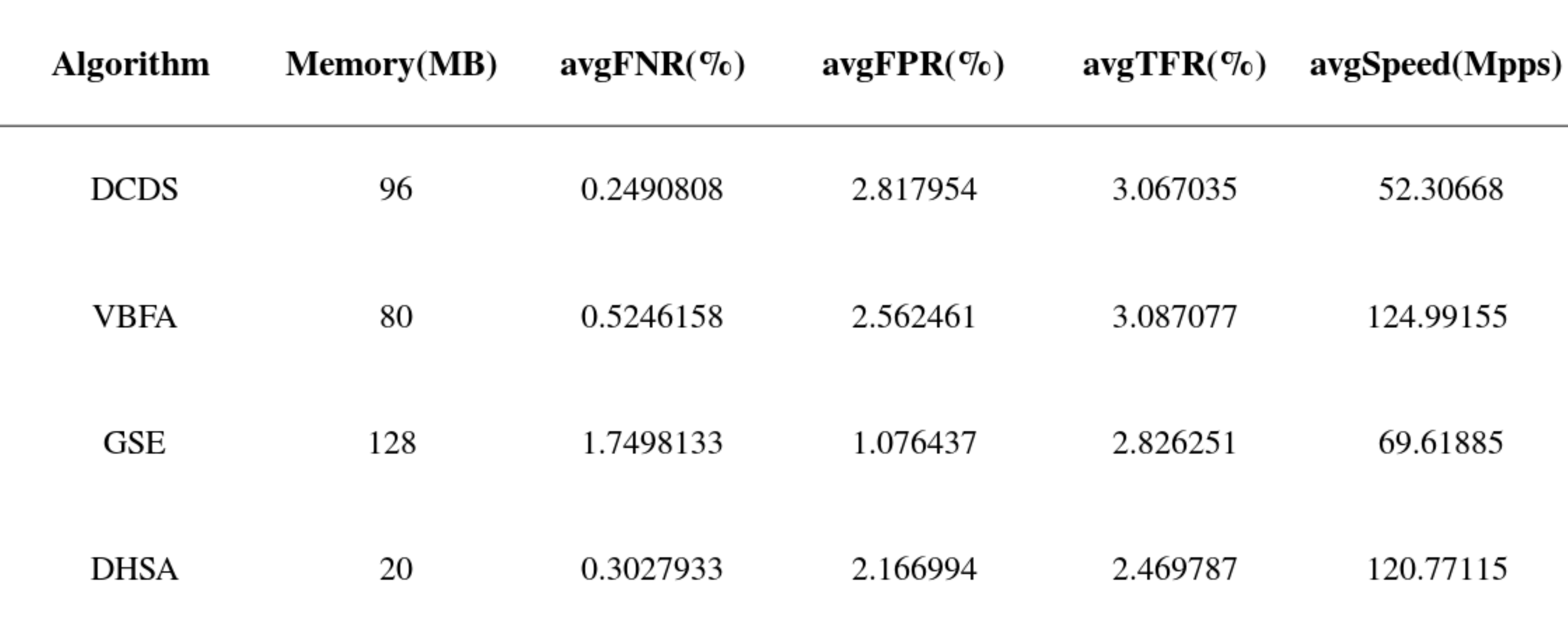}
\end{tabular}
\end{table*}

From table \ref{tbl_avg_hsd_rlt} we can see that, DHSA uses the smallest memory, smaller than one-fourth memory used by others, and has the lowest total false rate. DHSA also has a fast speed to deal with core networks traffic in real time.

Packets of our experiment traffic have average size of 800K bit. With this value we can using traffic size processing speed to replace packets process speed. From the perspective of traffic size, DHSA can dealing with 750 Gb/s traffics ($\frac{120*800*8}{1024}$) and VBFA can dealing with 781.25 Gb/s traffics ($\frac{125*800*8}{1024}$). Because host could be divided by their right bits into different part, and processing different parts separately. So all super point detection algorithms could be applied to higher speed traffic by increasing memory. 
From table \ref{tbl_avg_hsd_rlt} we can see that, VBFA requires 80 MB of memory to deal with traffic with 1 Mpps and DHSA requires only 20 MB of memory. For a GPU with 4 GB of memory, VBFA can process traffic with speed at most 51 Mpps (1 Mpps * 4096 MB/80 MB) and DHSA can process traffic with speed 204 Mpps (1 Mpps * 4096 MB/20 MB). Considering processing speed and memory requirement, DHSA has the best capability to cope with core network with bandwidth as high as 120 Mpps( or 750 Gb/s) on this cheap GPU.
\section{Conclusion}
GPU is a popular parallel platform and its price will grow with its global memory increase.  
A small memory requirement of an algorithm running on GPU will be economic. Our memory efficient super point detection algorithm DHSA is a such one that has the highest accuracy and using only smaller than one-fourth memory of others'. The excellent performance of DHSA comes from a novel designed double direction hash functions group DHFG. DHFG has a high randomness which makes DHSA can make full use of every bit in memory. When detecting super point, we not only need a high randomness functions but also the ability to restore super points from memory. Unlike other hash functions, DHFG can reconstruct hosts from the set of hashed values. Reversible, random and simple in operations, DHFG let DHSA become the most economical and efficient choice for core network's super point detection.

\iftoggle{ACM}{
\bibliographystyle{ACM-Reference-Format}
}

\bibliography{..//..//ref} 

\end{document}